\documentclass[a4paper,conference]{IEEEtran}

\IEEEsettopmargin{t}{30mm}
\IEEEquantizetextheight{c}
\IEEEsettextwidth{14mm}{14mm}
\IEEEsetsidemargin{c}{0mm}

\usepackage[utf8]{inputenc}
\usepackage[T1]{fontenc}
\usepackage{url}
\usepackage{ifthen}
\usepackage{cite}
\usepackage[cmex10]{amsmath} 
\usepackage{amssymb}
\usepackage{algorithm}
\usepackage{algpseudocode}
\usepackage{graphicx}
\usepackage{color}
\usepackage{xcolor}
\usepackage{cite}
\usepackage{subcaption}
\usepackage{comment}
\usepackage{mathtools}
\usepackage{tabulary}
\usepackage[nohyperlinks]{acronym}
\usepackage{upgreek}
\usepackage{placeins}
\usepackage[super]{nth}
\usepackage{dsfont}
\usepackage{hyperref}
\usepackage{array}
\usepackage{soul}
\usepackage{dsfont}
\usepackage{cancel}
\usepackage{multirow}
\usepackage{bm}
\usepackage{bbm}
\usepackage{amsthm}
\usepackage{url}
\usepackage[acronym]{glossaries}
\usepackage{balance}
\glsdisablehyper

\newtheorem{proposition}{\bf Proposition}

\newacronym{RIS}{RIS}{Reconfigurable Intelligent Surface}
\newacronym{NL}{NL}{NonLinear}
\newacronym{SIM}{SIM}{Stacked Intelligent Metasurfaces}
\newacronym{DNN}{DNN}{Deep artificial Neural Network}
\newacronym{NN}{NN}{neural network}
\newacronym{EI}{EI}{Edge Inference}
\newacronym{EM}{EM}{ElectroMagnetic}
\newacronym{TOC}{GOC}{Goal-Oriented Communications}
\newacronym{TX}{TX}{Transmitter}
\newacronym{RX}{RX}{Receiver}
\newacronym{UE}{UE}{User Equipment}
\newacronym{BS}{BS}{Base Station}
\newacronym{LoS}{LoS}{Line-of-Sight}
\newacronym{NLoS}{NLoS}{Non Line-of-Sight}
\newacronym{CNN}{CNN}{Convolutional Neural Network}
\newacronym{FFNN}{FFNN}{Feed-Forward Neural Network}
\newacronym{ReLU}{ReLU}{Rectified Linear Unit}
\newacronym{JSCC}{JSCC}{Joint Source Channel Coding}
\newacronym{E2E}{E2E}{End-to-End}
\newacronym{D2D}{D2D}{Device-to-Device}
\newacronym{DeepSC}{DeepSC}{Deep Semantic Communications}
\newacronym{SISO}{SISO}{Single Input Single Output}
\newacronym{MISO}{MISO}{Multiple Input Single Output}
\newacronym{MIMO}{MIMO}{Multiple-Input Multiple-Output}
\newacronym{OAC}{OAC}{Over-the-Air Computation}
\newacronym{MSE}{MSE}{Mean Squared Error}
\newacronym{SOTA}{SotA}{State of the Art}
\newacronym{SGD}{SGD}{Stochastic Gradient Descent}
\newacronym{IoT}{IoT}{Internet of Things}
\newacronym{AE}{AE}{Auto-Encoder}
\newacronym{PDF}{PDF}{Probability Density Function}
\newacronym{AWGN}{AWGN}{Additive White Gaussian Noise}
\newacronym{CSI}{CSI}{Channel State Information}
\newacronym{CFR}{CFR}{Channel Frequency Response}
\newacronym{ISAC}{ISAC}{Integrated Sensing and Communications}
\newacronym{iid}{i.i.d.}{independent and identically distributed}
\newacronym{MINN}{MINN}{Metasurfaces-Integrated Neural Network}
\newacronym{MS}{MS}{MetaSurface}
\newacronym{CE}{CE}{Cross Entropy}
\newacronym{HRIS}{HRIS}{Hybrid RIS}
\newacronym{RF}{RF}{Radio-Frequency}
\newacronym{VAE}{VAE}{Variational Auto-Encoder}
\newacronym{PSK}{PSK}{Phase Shift Keying}
\newacronym{SVD}{SVD}{Singular Value Decomposition}
\newacronym{MLP}{MLP}{Multi Layer Perceptron}
\newacronym{SNR}{SNR}{Signal-to-Noise Ratio}
\newacronym{MI}{MI}{Mutual Information}
\newacronym{MEC}{MEC}{Multi-access Edge Computing}
\newacronym{WMMSE}{WMMSE}{Weighted Minimum Mean Square Error}
\newacronym{ANN}{ANN}{Artificial Neural Network}
\newacronym{SLFN}{SLFN}{Single hidden Layer Feedforward Network}
\newacronym{ELM}{ELM}{Extreme Learning Machine}
\newacronym{LTI}{LTI}{Linear Time Invariant}
\newacronym{OTA}{OTA}{Over-The-Air}
\newacronym{AM}{AM}{Amplitude Modulation}
\newacronym{LS}{LS}{Least-Squares}
\newacronym{DC}{DC}{Direct Current}
\newacronym{RFC}{RFC}{Radio-Frequency Chain}
\newacronym{ML}{ML}{Machine Learning}
\newacronym{XL}{XL}{eXtremely Large}
\newacronym{PGD}{PGD}{Projected Gradient Descent}
\newacronym{WBCD}{WBCD}{Wisconsin Breast Cancer Dataset}
\newacronym{DMA}{DMA}{Dynamic Metasurface Antenna}
\newacronym{CMS}{CMS}{Cascaded MetaSurfaces}
\newacronym{GO}{GO}{Genetic Optimization}

\title{Over-The-Air Extreme Learning Machines with Nonlinear Stacked Intelligent Metasurfaces}

\newcommand\blfootnote[1]{%
  \begingroup
  \renewcommand\thefootnote{}\footnotetext{#1}%
  \addtocounter{footnote}{-1}%
  \endgroup
}

\author{%
  \IEEEauthorblockN{Kyriakos Stylianopoulos$^1$,
  Mattia Fabiani$^{2,3}$,
  Giulia Torcolacci$^{2,3}$,
  Davide Dardari$^{2,3}$, George C. Alexandropoulos$^{1}$
  }
  \IEEEauthorblockA{$^1$Department of Informatics and Telecommunications, National and Kapodistrian University of Athens, Greece} 
  \IEEEauthorblockA{$^2$DEI, University of Bologna, Italy; $^3$National Laboratory of Wireless Communications (WiLab), CNIT, Italy\\
    e-mails: \{kstylianop, alexandg\}@di.uoa.gr, \{mattia.fabiani5, g.torcolacci, davide.dardari\}@unibo.it
}
}

\begin{document}

\maketitle
\blfootnote{This work has been supported by the SNS JU projects 6G-DISAC and TIMES under the EU's Horizon Europe research and innovation program under grant agreement numbers 101139130 and 101096307, respectively. 
The work of G. Torcolacci was funded by an NRRP Ph.D. grant.}

\begin{abstract}
The recently envisioned goal-oriented communications paradigm requires machine learning inference to be performed directly on wirelessly transferred data. This paper presents an eXtremely Large (XL) Multiple-Input Multiple-Output (MIMO) system that operates as an Extreme Learning Machine (ELM) to execute Over-The-Air (OTA) binary classification. To reduce hardware complexity, the receiver is equipped with cascaded metasurfaces terminating in a single radio-frequency chain. A front metasurface layer applies a fixed nonlinear response to the incoming signal, acting as the ELM's activation function. Subsequent tunable linear metasurface layers physically approximate the trained network weights directly in the wave domain. Numerical evaluations across diverse datasets showcase that our XL MIMO architecture achieves classification accuracy comparable to idealized digital models, thereby proving the viability of low-complexity, wave-domain OTA learning.

\end{abstract}

\begin{IEEEkeywords}
Over-the-air inference, extreme learning machines, nonlinear signal processing, XL MIMO, SIM.
\end{IEEEkeywords}

\section{Introduction}
Future wireless networks will leverage \gls{EI} to jointly train transceiver pairs as end-to-end \gls{ML} models for efficient sensory data inference~\cite{GO_review2}.
By exchanging task-specific representations through the channel, \gls{EI} overcomes the inefficiencies of conventional decoupled designs in terms of data rate and computational burden, since feature extraction is performed alongside encoding at the \gls{TX}, while the \gls{RX} directly infers the target values instead of reconstructing the input data~\cite{GJZ24_SIM_TOC, Stylianopoulos_GO}.
 
To further improve computational efficiency, \gls{OTA} computing exploits the wireless propagation domain by performing computations directly through the superposition of traveling \gls{RF} signals~\cite{OTA_review}.
The \gls{OTA} paradigm has recently attracted interest in wireless \gls{ML} applications.
Specifically, hardware platforms based on \glspl{MS} have been proposed to emulate \gls{DNN} layers~\cite{XYN18_Diffractive_DNN, Momeni2022_wave_DL_computing,Stylianopoulos_GO,GJZ24_SIM_TOC} for \gls{OTA} inference, which are trained through backpropagation, or to approximate digitally trained \gls{DNN} weight matrices \gls{OTA}~\cite{Gunduz_Layer_Approximation, Pandolfo_SIM_semantic_alignment}.
Nevertheless, many existing systems still rely on digital processing and lack theoretical foundations.
Crucially, most \gls{MS}-based \gls{DNN} implementations are limited to performing linear operations~\cite{Styl:25}, which significantly constrains the approximation capability of the resulting ML models.

Addressing some of these limitations, the work of \cite{Stylianopoulos_MIMO_ELM} proposed an \gls{XL} \gls{MIMO} architecture that operates as an \gls{ELM}~\cite{HuaZhuSie:06}, enabling the partial execution of \gls{DNN} computations \gls{OTA}. In this framework, the wireless channel is exploited as a source of random hidden-layer weights, while the \gls{RX} analog combiner implements the output layer.
This approach ensures rapid training and efficient reconfiguration under channel variations, while retaining the universal function approximation capability. However, it faces practical limitations and scalability concerns due to real-valued signal constraints and hardware complexity induced by the use of \gls{NL} power amplifiers and numerous \gls{RF} chains. 

In this paper, capitalizing on the complex-domain \gls{ELM} framework~\cite{HuaLiCheSie:08} and building upon recent \gls{NL} metamaterial advancements~\cite{Styl:25,FabTorDar:25,Omr:25}, we present an \gls{XL}-\gls{MIMO}-\gls{ELM} system with an \gls{RX} structure comprising \gls{SIM}~\cite{AXN23_SIM}, followed by a single antenna and its respective \gls{RF} chain. The first \gls{MS} layer interfacing with the \gls{MIMO} channel is composed of unit cells exhibiting identical fixed \gls{NL} responses and implements the \gls{ELM} activation function, while the subsequent linear \gls{MS} layers realize trainable \gls{OTA} combining, effectively approximating the digital \gls{ELM} output weights. The proposed \gls{OTA}-\gls{ELM} system enables fast training and minimizes digital processing at the reception side, while significantly reducing the hardware complexity with respect to the \gls{XL}-\gls{MIMO}-\gls{ELM} of~\cite{Stylianopoulos_MIMO_ELM}. To configure the linear \gls{MS} layers, we propose two distinct \gls{OTA} training strategies. The first relies on \gls{PGD} to approximate the ideal \gls{LS} solution when the received signal at each \gls{NL} unit cell can be measured. Alternatively, we introduce a black-box \gls{GO} approach that only requires the final scalar output at the \gls{RF} chain, significantly reducing the measurement overhead. The proposed architecture is numerically validated on standard datasets, demonstrating that it approaches the accuracy of idealized digital models while being robust to system variations.


\section{The Proposed XL MIMO System Model} \label{sec:system-model}

Consider a narrowband \gls{XL} \gls{MIMO} system with an $N_{\rm t}$-antenna \gls{TX} and an \gls{RX} comprising a diffractive receiving \gls{MS}, followed by an analog combining stage (detailed subsequently) and a single reception \gls{RF} chain.
Instead of performing conventional wireless communications, the system is trained end-to-end to act as an \gls{OTA} function approximator~\cite{Stylianopoulos_MIMO_ELM}.
In particular, given a prior dataset $\mathcal{D} \triangleq \{(\mathbf{x}^{(i)}, z^{(i)})\}_{i=1}^D$ of $D$ input-target pairs, the \gls{XL} \gls{MIMO} system is designed to approximate the $\mathbf{x} \to z$ mapping so that the \gls{TX} observes the input data $\mathbf{x}$ (not necessarily belonging to $\mathcal{D}$) and the \gls{RX} estimates its (unobserved) target value $z$. The system is thus intended to perform \gls{EI}, with all computational processing performed exclusively \gls{OTA}, i.e., in the analog/\gls{RF} domain.

We assume that $\mathbf{x}^{(i)} \in (0,1)^{N_{\rm t}}$ and, without loss of generality, $z^{(i)} \in \{0,1\}$, that is, the dimension of the data observations are equal to the number of \gls{TX} antennas and the dataset is used for real-valued binary classification. Note that the \gls{TX} may incorporate a trainable feature extraction module to reduce the dimensionality of $\mathbf{x}$~\cite{Stylianopoulos_GO, GJZ24_SIM_TOC}; however, this aspect is left for future investigation.
For reasons that will be clarified later, the data symbols are subject to \gls{AM}; accordingly, each element of the transmitted signal $\bar{\mathbf{x}} \in \mathbb{C}^{N_{\rm t} \times 1}$ is defined $\forall i=1,\dots,N_{\rm t}$ as follows:
\begin{equation}\label{eq:transmit-signal}
    [\bar{\mathbf{x}}]_i \triangleq [\mathbf{x}]_i \exp(\jmath \pi  [\bm \psi]_i),
\end{equation}
where the elements of $\bm \psi$ may be chosen arbitrarily.
The baseband representation of the impinging signal at the \gls{MS} layer of the \gls{RX}, which is composed of $N_{\rm r}$ metamaterial elements, can be expressed as (baseband representation):
\begin{equation}\label{eq:received-signal}
    \mathbf{y} \triangleq \mathbf{H} \bar{\mathbf{x}} \in \mathbb{C}^{N_{\rm r} \times 1},
\end{equation}
where $\mathbf{H} \in \mathbb{C}^{N_{\rm r} \times N_{\rm t}}$ represents the \gls{XL} \gls{MIMO} channel response.
For the subsequent theoretical analysis, we consider the case where $\mathbf{H}$ follows a Ricean fading distribution~\cite{6184250} and remains quasi-static for the duration of the training process:
\begin{equation}\label{eq:ricean}
   \mathbf{H} = \sqrt{P_L} \left( \sqrt{\frac{K}{1+K}} \mathbf{H}_{\rm LoS} + \sqrt{\frac{1}{1+K}} \mathbf{H}_{\rm NLoS} \right),
\end{equation}
where $\mathbf{H}_{\rm LoS}$ is a rank-$1$ matrix of steering modeling for the \gls{LoS} component, while $\mathbf{H}_{\rm NLoS} \sim \mathcal{CN}(\mathbf{0}, 1/\sqrt{N_{\rm t} N_{\rm r}}\mathbf{I})$, and $K$ is the Ricean factor that controls the dominance of either component.

We consider a diffractive \gls{MS} layer at the \gls{RX} composed of unit elements applying a memoryless \gls{NL} transformation. Denoting with $F(\cdot)$ the bandpass response of the generic element of the \gls{NL} \gls{MS}, the baseband-equivalent output $g(\cdot)$ preserves the phase of the input, while transforming its envelope through the first-order harmonic extraction. The resulting element-wise mapping of the \gls{MS} is expressed as $g(\mathbf{y}) \triangleq C(|\mathbf{y}|) \exp(\jmath {\rm arg}\{\mathbf{y}\})$, where $C(\cdot)$ denotes the \gls{AM}/\gls{AM} characteristic derived as follows~\cite{FabTorDar:25}:
\begin{equation} \label{eq:C}
C(v)=\frac{2}{\pi} \int_{0}^{\pi} F(v\cos(\phi)) \cos(\phi) \text{d}\phi.
\end{equation}
In this paper, we consider an element-wise thresholding device characterized overall by the positive bias $\mathbf{b} \in \mathbb{R}^{N_{\rm r} \times 1}_+$, whose elements are drawn from an appropriate distribution during fabrication, hence ensuring low complexity.
From \eqref{eq:C}, the transform yields the following piecewise mapping for each element $j=1,\dots,N_{\rm r}$ of the diffractive MS layer:
\begin{equation}
    C(|[\mathbf{y}]_j|) = \begin{cases} 
    \mathbf{0}, & |[\mathbf{y}]_j| \le [\mathbf{b}]_j \\
    \begin{aligned}
    \frac{1}{\pi} \left( |[\mathbf{y}]_j| \arccos\left(\frac{[\mathbf{b}]_j}{|[\mathbf{y}]_j|}\right)\right. & \\
    \left.- [\mathbf{b}]_j \sqrt{1 - \left(\frac{[\mathbf{b}]_j}{|[\mathbf{y}]_j|}\right)^2}\right) & ,
    \end{aligned} & |[\mathbf{y}]_j| > [\mathbf{b}]_j
    \end{cases}.
    \label{eq: C[|y|]}
\end{equation}
While this expression captures the exact physical behavior of the \gls{MS} elements, the transcendental terms are computationally demanding for practical optimization.
By approximating the transition for $|[\mathbf{y}]_j| > [\mathbf{b}]_j$ as a continuous quasi-linear slope, the response of the \gls{MS} layer can be expressed as the well-known softplus activation function~\cite{Wiemann2024_softplus}:
\begin{equation}\label{eq:softplus-activation}
    g([\mathbf{y}]_j) \simeq 0.2
    \log\left( 1+e^{5 (|[\mathbf{y}]_j| -[\mathbf{b}]_j)  }\right)
    \exp\left(j {\rm arg}\{[\mathbf{y}]_j\}\right).
\end{equation}
Once processed through~\eqref{eq:softplus-activation}, the resulting signal is linearly combined via the controllable weight vector $\mathbf{w} \in \mathbb{C}^{N_{\rm r} \times 1}$ prior to its feeding to the \gls{RF} chain, yielding the scalar output:
\begin{equation}\label{eq:output}
    \hat{z} \triangleq \mathbf{w}^\top g(\mathbf{y}) + \tilde{n},
\end{equation}
where $\tilde{n} \in \mathbb{C}$ denotes the \gls{AWGN} introduced at the RX. Herein, the physical implementation of~\eqref{eq:output} is based on \gls{SIM} and is detailed in Section~\ref{sec:sim}.

\subsection{XL MIMO as an ELM} 
\label{sec:xl-mimo-elm-theory}
The previously presented \gls{XL} \gls{MIMO} system may be regarded as a form of an \gls{ELM}, where the transformations~\eqref{eq:received-signal} and~\eqref{eq:softplus-activation} implement the random hidden layer (with random coefficients ${\bm \theta} \triangleq \{ \mathbf{H}, \mathbf{b} \}$) and the combining weights $\mathbf{w}$ of~\eqref{eq:output} play the role of the trainable weights of the output layer.
We thus leverage the developed mathematical framework~\cite{HuaZhuSie:06, LiHuaSarSun:05, Stylianopoulos_MIMO_ELM} for \gls{ML} inference, which accounts for random parameters alongside trainable weights, enabling both the training procedure and its theoretical guarantees to be rigorously described.
To find the optimal weight vector for a given dataset $\mathcal{D}$, let us first define the \gls{ELM} activation matrix $\mathbf{G} \in \mathbb{C}^{D \times N_{\rm r}}$ as the transpose of the activated signals at the aforedescribed \gls{NL} \gls{MS} layer at the RX front:
\begin{equation}\label{eq:hidden-layer-outputs}
    \mathbf{G} \triangleq [g(\mathbf{y}^1), \dots, g(\mathbf{y}^D)]^{\top}.
\end{equation}
By further denoting the vector of target values for the whole dataset as $\mathbf{z} \triangleq [z^1 , \dots , z^D]^\top \in \mathbb{C}^{D \times 1}$, $\mathbf{w}$ can be optimized to minimize the \gls{LS} error between the target and output values, following the standard \gls{ELM} formulation:
\begin{equation}\label{eq:LS-objective}
    \mathbf{w}^{*} \triangleq \arg \min_{\mathbf{w}} \| \mathbf{z} -  \mathbf{G}\mathbf{w} \|_{2}^{2}.
\end{equation}
This yields the closed-form solution:
\begin{equation} \label{eq:weights-opt}
    \mathbf{w}^{*} = \left(\mathbf{G}^{\rm H} \mathbf{G} + \ell \mathbf{I}\right)^{-1}\mathbf{G}^{\rm H} \mathbf{z} \in \mathbb{C}^{N_{\rm r} \times 1},
\end{equation}
which accounts for L2 regularization, controlled by the hyperparameter $\ell>0$, to ensure generalization beyond the training dataset $\mathcal{D}$. 
A key advantage of \glspl{ELM} is their universal approximation property, expressed for the proposed \gls{ELM} as follows (a full proof with additional analysis has been deferred for the journal version of this work):

\begin{proposition}
Consider a dataset $\mathcal{D} \triangleq \{(\mathbf{x}^{(i)}, z^{(i)})\}_{i=1}^D$ and the system defined by~\eqref{eq:transmit-signal},~\eqref{eq:received-signal},~\eqref{eq:softplus-activation}, and~\eqref{eq:output}.  
In the high-\gls{SNR} regime ($\tilde{n} \to 0$), assuming 
$\mathbf{H}$ follows the Ricean model in~\eqref{eq:ricean} and 
$\mathbf{b}$ is drawn from a continuous distribution with positive support, there exists a weight vector 
$\mathbf{w}^\ast$ such that 
$\mbox{$\|\mathbf{z} \!-\! \mathbf{G} \mathbf{w}^\ast\|_2^2 = 0$}$ for $N_{\rm r} \!=\! D$ with probability $1$.
\end{proposition}

\begin{proof}[Sketch of proof]
To show that~\eqref{eq:LS-objective} admits a zero-error solution, it suffices to show that $f_{\mathcal{D}}(\theta) \triangleq {\rm det}(\mathbf{G}) \neq 0$ with probability $1$.
To that end, the following corollary of the Identity Theorem for real analysis can be used, stating that, for a real analytic function $f_{\mathcal{D}}(\theta)$ on an open connected domain $\mathcal{I}$ of the isomorphic space $\mathbb{R}^{2(N_{\rm r}D + N_{\rm r})}$ in which $f_{\mathcal{D}}(\theta)$ is not identically zero, the zero set has zero Lebesgue measure~\cite{Mityagin2020_Zero_Set_Identity_Theorem}.
Since it consists of a composition of real analytic functions, 
$f_{\mathcal{D}}(\theta)$ is real analytic, as long as $|[ \mathbf{H}\bar{\mathbf{x}}]_j| \neq 0$ $\forall j=1,\dots,N_r$. Such $\mathbf{H}$ values have zero probability of occurring under Ricean fading, since they are embedded in a lower-dimensional manifold. Moreover, they form a closed set with co-dimension compared to the full parameter space of $2$, therefore, excluding it, leaves the domain of $f_{\mathcal{D}}(\theta)$ connected and open. Finally, $f_{\mathcal{D}}(\theta)$ is not identically zero, as shown by the counterexample in which $\mathfrak{Im}(\mathbf{H})=\mathbf{0}$. By performing \gls{AM} transmission of $\mathbf{x}^{(i)}$ with ${\bm \psi} = \mathbf{0}$, the proposed framework reduces to a standard real-valued \gls{ELM}, for which, $\mathbf{G}$ is invertible~\cite[Theorem 2.2]{HuaZhuSie:06}.
\end{proof}

\subsection{OTA Analog Combining using SIM}\label{sec:sim}
Up to this point, the physical implementation of the combining operation has been left unspecified to focus on theoretical analysis. 
We now describe the \gls{RX} architecture that performs this combining \gls{OTA}. 
Specifically, we consider that the outputs $g(\mathbf{y})$ of the diffractive \gls{NL} \gls{MS} layer are fed into a \gls{SIM} of $L$ diffractive linear \glspl{MS}.
Each $l$-th layer ($l=1,\dots, L$) comprises a square grid of $N_l$ elements spaced by $\lambda/ 2$, with $\lambda = c/f_0$ denoting the wavelength at the carrier frequency $f_0$, and $c$ is the speed of light.
Let ${\bm \Omega}_l \in \mathbb{C}^{N_l \times N_{l-1}}$ denote the signal propagation coefficients between the $(l-1)$-th and the $l$-th \gls{MS} layers, where $N_0 = N_{\rm r}$ indicates the number of elements of the aforedescribed \gls{NL} \gls{MS}, and ${\bm \omega}_L \in \mathbb{C}^{N_L \times 1}$ represents the propagation between the last \gls{MS} layer and the single-antenna element attached to the \gls{RX} \gls{RF} chain. Typical works leveraging \gls{SIM} technology~\cite{AXN23_SIM} assume free-space propagation between the \gls{MS} layers, i.e., considering an anechoic enclosure, and model the element-to-element propagation through geometric optics~\cite{AXN23_SIM, XYN18_Diffractive_DNN, GJZ24_SIM_TOC, Momeni2022_wave_DL_computing}.
In this work, we model each ${\bm \Omega}_l$ as a full-rank pseudo-random matrix~\cite{Styl:25}, which presumes a reverberating enclosure featuring a non-uniform, non-planar distribution of elements~\cite{Tacid_PhysFad, FILM_25}.
Arguably, this choice is more realistic as it accounts for multipath components arising from imperfections of the enclosure and allows for the \gls{MS} layers to be placed arbitrarily close.
Moreover, the richer propagation diversity compared to geometric optics provides substantial gains for the optimization framework.

The responses of each $l$-th \gls{MS} layer are expressed as ${\bm \phi}_l \triangleq {\bm \alpha}_l \exp(\jmath \pi {\bm \varphi}_l)$, with the amplitudes ${\bm \alpha}_l \in [0,1]^{N_l \times 1}$ and the phase shifts ${\bm \varphi}_l \in [0,2]^{N_l \times 1}$ being {\em controllable} parameters.
Setting ${\bm \Phi}_l \triangleq {\rm diag}({\bm \phi}_l)$ and ${\bm \varphi} \triangleq \{{\bm \phi}_l \}_{l=L-1}^1$, the overall transfer function of the $L$ linear \glspl{MS} is given by (as in Fig.~\ref{fig:system-model}):
\begin{equation}\label{eq:sim-response}
    \mathbf{w}_{\bm \varphi} \triangleq \Big( {\bm \omega}_L^\top \prod_{l=L-1}^{1} {\bm \Phi}_{l} {\bm \Omega}_{l}\Big)^\top \in \mathbb{C}^{N_{\rm r} \times 1}.
\end{equation}
Thus, the \gls{SIM} response is used to perform \gls{OTA} combining by substituting $ \mathbf{w}$ in~\eqref{eq:output} with $\mathbf{w}_{\bm \varphi}$.

\begin{figure}
    \centering
    \includegraphics[width=1.05\linewidth]{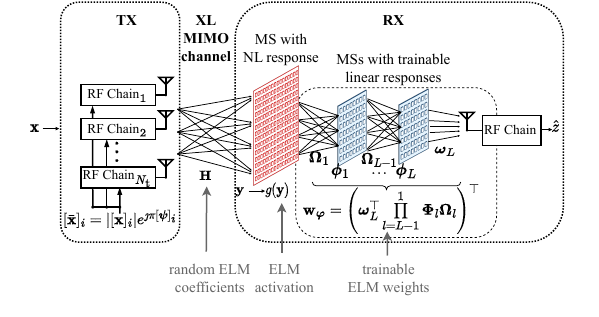}
    \caption{The proposed \gls{XL} \gls{MIMO} system for implementing the proposed \gls{OTA}-\gls{ELM} framework using linear and \gls{NL} \glspl{MS}. The channel and \gls{MS} responses are used as components of the \gls{ELM} algorithm to realize \gls{OTA} inference. The flow of computation during the forward pass is also sketched in this figure.}
    \vspace{-0.3cm}
    \label{fig:system-model}
\end{figure}

\begin{figure*}[t]
    \centering
    \begin{subfigure}[b]{0.3\textwidth}
        \centering
        \includegraphics{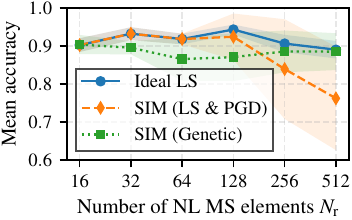}
        \caption{WBCD.}
        \label{fig:results_breast_cancer}
    \end{subfigure}
    \hfill
    \begin{subfigure}[b]{0.3\textwidth}
        \centering
        \includegraphics{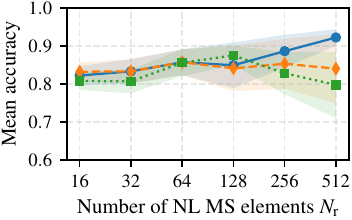}
        \caption{Parkinson's.}
        \label{fig:results_parkinsons}
    \end{subfigure}
    \hfill
    \begin{subfigure}[b]{0.3\textwidth}
        \centering
        \includegraphics{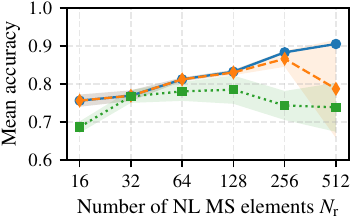}
        \caption{MNIST.}
        \label{fig:results_mnist}
    \end{subfigure}
    \caption{Classification accuracy the proposed \gls{OTA}-\gls{ELM} approaches versus the number of metamaterials $N_{\rm r}$ (corresponding to the number of trainable parameters) at the \gls{RX}'s \gls{NL} \gls{MS} layer, considering three distinct datasets.}
    \vspace{-0.5cm}
    \label{fig:overall_results}
\end{figure*}

\section{Training Approaches with a SIM} 
\label{sec:training}

\subsection{Approximation of the LS Solution}
Since the optimal digital weights $\mathbf{w}^\ast$ are known from~\eqref{eq:weights-opt}, the \gls{SIM} responses can be optimized so that $\mathbf{w}_{\bm \varphi}$ closely approximates $\mathbf{w}^\ast$, leading to the following optimization problem with respect to the \gls{SIM} parameters:
\begin{equation}\label{eq:approx-objective}
    {\bm \varphi}^{\rm GD} \triangleq \{ {\bm \phi}^\ast_l \}_{l=1}^L \triangleq \arg \min_{{\bm \varphi}} \| \mathbf{w}^\ast -\mathbf{w}_{\bm \varphi}\|_{2}^{2} \, ,
\end{equation}
which can be solved via \gls{PGD}.
Since automatic differentiation may be applied, and due to a lack of space, the gradient derivations are omitted. However, we note that this approach relies on the knowledge of $\mathbf{G}$~\cite{Stylianopoulos_MIMO_ELM}, which involves measuring the impinging signal at each element of the front \gls{NL} \gls{MS} while using a single \gls{RF} chain.
To this end, a time-division measurement procedure can be employed, where, at each time slot $t=1,\dots,N_{\rm r}$, only the $t$-th metamaterials is active while all others are turned off (equivalently, setting $[\mathbf{w}]_j$ or $[\mathbf{w}_{\bm \varphi}]_j$ to $\mathbbm{1}_{j=t}$ in~\eqref{eq:output}).
Since \glspl{MS} can be reconfigured with sub-$\mu$s latency, $\mathbf{G}$ can be measured, and the optimization can be performed within the channel coherence time, provided the channels exhibit reasonably slow fading.
Furthermore, using multiple \gls{RX} \gls{RF} chains solely for training allows parallel measurements, resulting in a linear speedup.
More details on the hardware realization of this approach will be provided in the journal version of this work.

\subsection{Direct SIM Training via Genetic Optimization}
Alternatively, we further apply a black-box optimization approach that does not rely on the observation of $\mathbf{G}$.
Specifically, we employ a \gls{GO} approach as follows. Starting from a population of $k_{\rm pop}$ candidate weight vectors ${\bm \varphi}^{\rm GO} \triangleq \{ {\bm \phi}_l \}_{l=1}^L$, we evaluate the following maximization objective:
\begin{equation}\label{eq:GO-objective}
    \mathcal{L}({\bm \varphi}^{\rm GO}) \triangleq - \Big\| \mathbf{z} -  \underbrace{\mathbf{G} \Big( {\bm \omega}_L^\top \prod\nolimits_{l=L-1}^{1} {\bm \Phi}_{l} {\bm \Omega}_{l}\Big)^\top}_{\hat{\mathbf{z}}\triangleq [\hat{z}^1 , \dots , \hat{z}^D]^\top \in \mathbb{C}^{D \times 1}}  \Big\|_{2}^{2}.
\end{equation}
The $k_{\rm top}$ top-performing candidates are admissible for offspring generation.
To create each offspring through uniform crossover, two parents are selected through random $3$-way tournaments.
Mutations are applied with $\Delta=\pm 0.2 \pm \jmath0.2$ to each  $[{\bm \varphi}^{\rm GO}]_j$.
The process is repeated for $k_{\rm gen}$ generations, after which the best ${\bm \varphi}^{\rm GO}$ is kept; therefore, $k_{\rm gen}\times k_{\rm pop}$ total evaluations of the objective are needed. This makes the latency of this approach also dependent on the switching time of the \gls{SIM} layers configuration.
Note, however, that only the output $\hat{\mathbf{z}}$ at the \gls{RF} chain is required, rather than the full matrix $\mathbf{G}$.

\begin{figure}[t]
    \centering
    \includegraphics[width=0.75\linewidth]{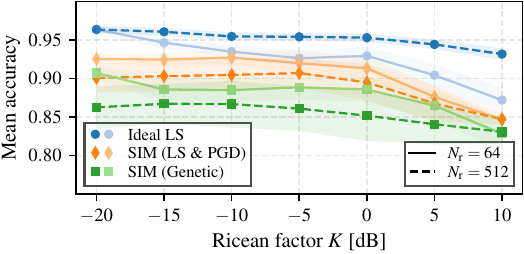}
    \caption{Classification accuracy of the \gls{OTA}-\gls{ELM} approaches over different Ricean factors, considering the \acrshort{WBCD} dataset.}
    \vspace{-0.3cm}
    \label{fig:results-ricean}
\end{figure}

\begin{figure}[t]
    \centering
    \includegraphics[width=0.75\linewidth]{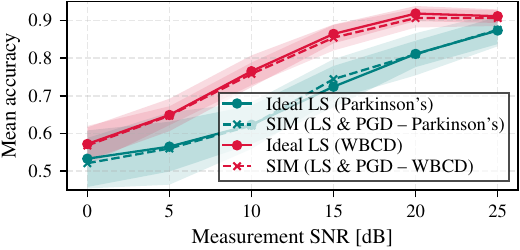}
    \caption{Performance of the ideal \gls{LS} benchmark and the approximate \gls{SIM} solution trained through \gls{PGD} versus different levels of the measurement \gls{SNR}, introduced by the observation of the received signals to compute the \gls{LS} solution.}
    \vspace{-0.3cm}
    \label{fig:results-snr}
\end{figure}

\section{Numerical Results and Discussion}
We have evaluated the proposed EI approaches on several standard small-to-medium-size binary classification datasets under a range of system settings and channel conditions. The considered datasets are the Parkinson's and the \gls{WBCD} of $22$ and $30$ numerical features (corresponding to $N_{\rm t}$ values) with $240$ and $700$ datapoints for disease diagnosis~\cite{UCI}, respectively, as well as the MNIST dataset of handwritten digit image recognition~\cite{MNIST} (of $6\times 10^4$ points subsampled to $100$ pixels and converted the objective to even/odd digit classification). Per-feature standardization was applied during preprocessing. To determine the predicted class, we have set a threshold of $0.5$ on $\mathfrak{Re}(\hat{z})$ of~\eqref{eq:output}, and we report in the performance figures that follow the mean and standard deviation of accuracy values over $100$ random initializations. For convenience, the phases ${\bm \psi}$ of the \gls{AM} signals were set identically to $\mathbf{0}$, and $\mathbf{b}$ was sampled from a Rayleigh distribution with scale parameter $\mathbb{E}[\|\mathbf{H}\|_{\rm F}]/(2 N_{\rm r} N_{\rm t})$, for it to be in the same order as $|\mathbf{y}|$.

Unless otherwise specified, the channel parameters were set as $P_L = -50$ dB and $K=0$ dB.
Since $\hat{z}$ contains one bit of information, the receive \gls{SNR} has negligible effect and was therefore set to $15$ dB.
For the \gls{SIM}, we employed $L=2$ linear layers, each consisting of $ 32 \times 32$ diffractive elements.
We considered ${\bm \Omega}_l \sim \mathcal{CN}(\mathbf{0}, P'_L \mathbf{I})$ with $P'_L$ set to $-10$ dB.
The regularization weight was set as $\ell=10^{-6}$ to avoid severe overfitting, and we allowed a maximum of $T=2000$ iterations of the employed \gls{PGD} procedure with step size $0.01$, irrespective of the number of training parameters for fairness, although convergence was achieved far earlier for most scenarios.
For the presented \gls{GO} approach, we set $k_{\rm pop} =100$, $k_{\rm gen}=200$, and $k_{\rm top}=30$.
The idealized \gls{LS} solution (which can be seen as an extension of the proposed framework of~\cite{Stylianopoulos_MIMO_ELM}) was used as an upper bound.
This idealized weighting can be realized either with digital (requiring $N_{\rm r}$ \gls{RF} chains) or analog (where a single \gls{RF} chain suffices) combining. For the latter case, phase shifters~\cite{Stylianopoulos_MIMO_ELM} or \gls{MS} structures with lossless waveguides~\cite{shlezinger2019dynamic} may be used.

The performance of the proposed training approaches for an increasing number of elements $N_{\rm r}$ at the \gls{RX}'s front \gls{MS} layer is displayed in Fig.~\ref{fig:overall_results}.
Note that $N_{\rm r}$ also corresponds to the number of trainable \gls{ELM} parameters. It is observed that, as $N_{\rm r}$ increases, the classification accuracy generally improves (except for slight overfitting in the simplest WBCD case), a trend that is consistent with the theoretical analysis as $N_{\rm r} \to D$.  
Both approximate \gls{SIM} solutions achieve performance close to the ideal case; however, for the largest $N_{\rm r}$ values, a noticeable degradation occurs due to an insufficient number of training iterations. Indicatively, the execution of Ideal LS, LS \& PGD, and GO algorithms takes, respectively, $0.32$, $6.21$, $178.5$ secs for the MNIST case when $N_{\rm r}=64$ and $2.37$, $10.28$, $242.18$ secs for $N_{\rm r}=512$ using unoptimized Python code running on a standard desktop computer.

In the previous investigation, we have assumed perfect knowledge of $\mathbf{G}$ during training. However, measuring $g(\mathbf{y})$ requires one or more \gls{RF} chains, which introduces measurement noise.  
To model this, we have included an \gls{AWGN} term $\tilde{\mathbf{n}}_0 \sim \mathcal{CN}(\mathbf{0}, 1/R\,\, {\rm diag}(|\mathbf{y}|^2))$ in~\eqref{eq:received-signal}, present only during the collection of signals for computing $\mathbf{G}$ and the \gls{LS} solution in~\eqref{eq:weights-opt}, in which $R$ denotes the target measurement \gls{SNR} in linear scale.
Figure~\ref{fig:results-snr} reports the performance across different measurement \gls{SNR} levels. It is shown that, while performance degrades at low measurement \gls{SNR}, the \gls{SIM} approximation achieves results comparable to the ideal \gls{LS} case.  
At sufficiently high measurement \gls{SNR}, both approaches converge to the noise-free performance, validating the idealized assumptions of our theoretical analysis.

Similarly, so far, we have assumed rich scattering ($K=0$) in the \gls{TX}-\gls{RX} channel, ensuring that $\mathbf{H}$ is a full-rank matrix. This improves the universal approximation capability of our framework, as the condition number of $\mathbf{G}$ decreases, yielding better conditioning for the \gls{LS} solution. Figure~\ref{fig:results-ricean} illustrates the impact of channel scattering, considering various Ricean factors. As depicted, both the ideal and \gls{SIM}-approximate methods maintain stable performance under sufficiently diverse channels. However, as the channel becomes \gls{LoS}-dominant, classification accuracy degrades because the columns of $\mathbf{G}$ become linearly dependent, and therefore, the conditions for universal approximation are no longer satisfied.


\vspace{-0.1cm}
\section{Conclusion}
This paper showcased that \gls{XL} \gls{MIMO} systems with properly designed \gls{MS} components at the \gls{RX} can perform computations equivalent to a complex-valued \gls{ELM}, enabling fully \gls{OTA} \gls{ML}-based inference on the transmitted data. The channel coefficients serve as random hidden-layer weights, while the \gls{RX}'s \gls{SIM} implements the activation function, bias, and trainable weights.  
Theoretical analysis confirmed universal approximation. Two practical approaches, one approximating the ideal \gls{LS} solution and the other based on \gls{GO} without requiring per-element signal observations, achieved performance close to the ideal \gls{LS} case under a range of system parameters.

We further note that the novelty of the proposed architecture lies in the joint use of a fixed NL-MS layer to physically realize the ELM activation and of a cascaded linear SIM to implement the trainable output layer within a single-RF-chain receiver, a combination not addressed by prior XL-MIMO-ELM~\cite{Stylianopoulos_MIMO_ELM} (which requires multiple analog beamformers and lacks an NL layer). Future work will provide a comparative analysis further positioning the proposed NL-SIM architecture against digital and over-the-air ELM implementations, thereby clarifying the novelty of the proposed solution and its associated benefits in terms of hardware complexity, energy consumption, and training/measurement overhead.

\vspace{-0.1cm}
\balance

\bibliographystyle{IEEEtran}
\bibliography{Biblio}

\end{document}